\documentclass{article}
\usepackage[english]{babel}
\usepackage[utf8x]{inputenc}

\usepackage[a4paper,top=2cm,bottom=2cm,left=3cm,right=3cm,marginparwidth=1.75cm]{geometry}
\usepackage{amsmath,amssymb,amsthm}
\usepackage{graphicx}
\usepackage{float}
\usepackage{bm}
\usepackage{tikz}
\usepackage[numbers,sort]{natbib}
\usepackage{algorithm,algorithmic,float}
\usepackage{enumitem}
\usepackage{comment}
\usepackage{mathtools}
\usepackage{url}
\usepackage{xcolor}
\usepackage{hyperref}
\hypersetup{
    colorlinks,
    linkcolor={blue!90!black},
    citecolor={blue!90!black},
    urlcolor={blue!90!black}
}
\usepackage{cleveref}
\usepackage{times}

\usepackage{makecell}
\usepackage{multirow}
\usepackage{nicematrix}
\usepackage{xfrac}

\usepackage{pmboxdraw}

\makeatletter
\renewcommand{\algorithmiccomment}[1]{\hfill{\color{blue!90!black}//~#1}}
\makeatother

\makeatletter
\newcommand*\rel@kern[1]{\kern#1\dimexpr\macc@kerna}
\newcommand*\widebar[1]{%
  \begingroup
  \def\mathaccent##1##2{%
    \rel@kern{0.8}%
    \overline{\rel@kern{-0.8}\macc@nucleus\rel@kern{0.2}}%
    \rel@kern{-0.2}%
  }%
  \macc@depth\@ne
  \let\math@bgroup\@empty \let\math@egroup\macc@set@skewchar
  \mathsurround\z@ \frozen@everymath{\mathgroup\macc@group\relax}%
  \macc@set@skewchar\relax
  \let\mathaccentV\macc@nested@a
  \macc@nested@a\relax111{#1}%
  \endgroup
}
\makeatother

\newtheorem{lemma}{Lemma}

\newtheorem{proposition}{Proposition}
\newtheorem{corollary}{Corollary}

\usepackage{braille}
\theoremstyle{definition}
\newtheorem{definition}{Definition}
\newtheorem{example}{Example}
\newtheorem{remark}{Remark}
\newtheorem{notation}{Notation}

\newcommand{\grobner}{Gr{\"o}bner}

\newcommand{\lm}{\mathrm{lm}}
\newcommand{\Mac}{\mathrm{Mac}}

\usepackage[dvipsnames]{xcolor}
\usepackage{booktabs}

\usepackage{tikz}
\usetikzlibrary{cd,matrix,shapes,decorations.pathreplacing,backgrounds,positioning,calc,shapes.misc,mindmap, arrows.meta,decorations,calligraphy}
\tikzset{
    lablvert/.style={anchor=south, rotate=90, inner sep=.5em}
}

\usepackage{subcaption}

\definecolor{vlgrey}{HTML}{707070}

\usepackage[textsize=tiny,colorinlistoftodos]{todonotes}

\newcommand{\Sasha}[2][]{}

\newcommand{\cU}{
\mathcal{U}
}

\newcommand{\rank}{
\operatorname{rank}
}

\title{Probably correct row echelon form in the F4 algorithm\thanks{This work has been supported by an ERC-2023-ADG grant for the ODELIX project (number 101142171) and partially supported by the NSF grant CCF-2212460.}}
\author{Alexander Demin\thanks{Laboratoire d'informatique de l'École polytechnique (LIX, UMR 7161), CNRS, École polytechnique, Institut Polytechnique de Paris, Palaiseau, France (email: \href{mailto:demin@lix.polytechnique.fr}{demin@lix.polytechnique.fr})}}
\date{\today}

\begin{document}

\maketitle

\makeatletter
\begingroup
\renewcommand{\thefootnote}{}
\renewcommand{\@makefntext}[1]{\noindent #1}

\footnotetext{%
\vspace*{0.6em}%
\begin{minipage}[t]{0.85\linewidth}
\vspace{0pt}%
\itshape
Funded by the European Union. Views and opinions expressed are however
those of the author(s) only and do not necessarily reflect those of the
European Union or the European Research Council Executive Agency.
Neither the European Union nor the granting authority can be held
responsible for them.
\end{minipage}%
\hfill
\begin{minipage}[t]{0.13\linewidth}
\vspace{0.20em}% move logo slightly down
% \centering
\includegraphics[height=0.80cm]{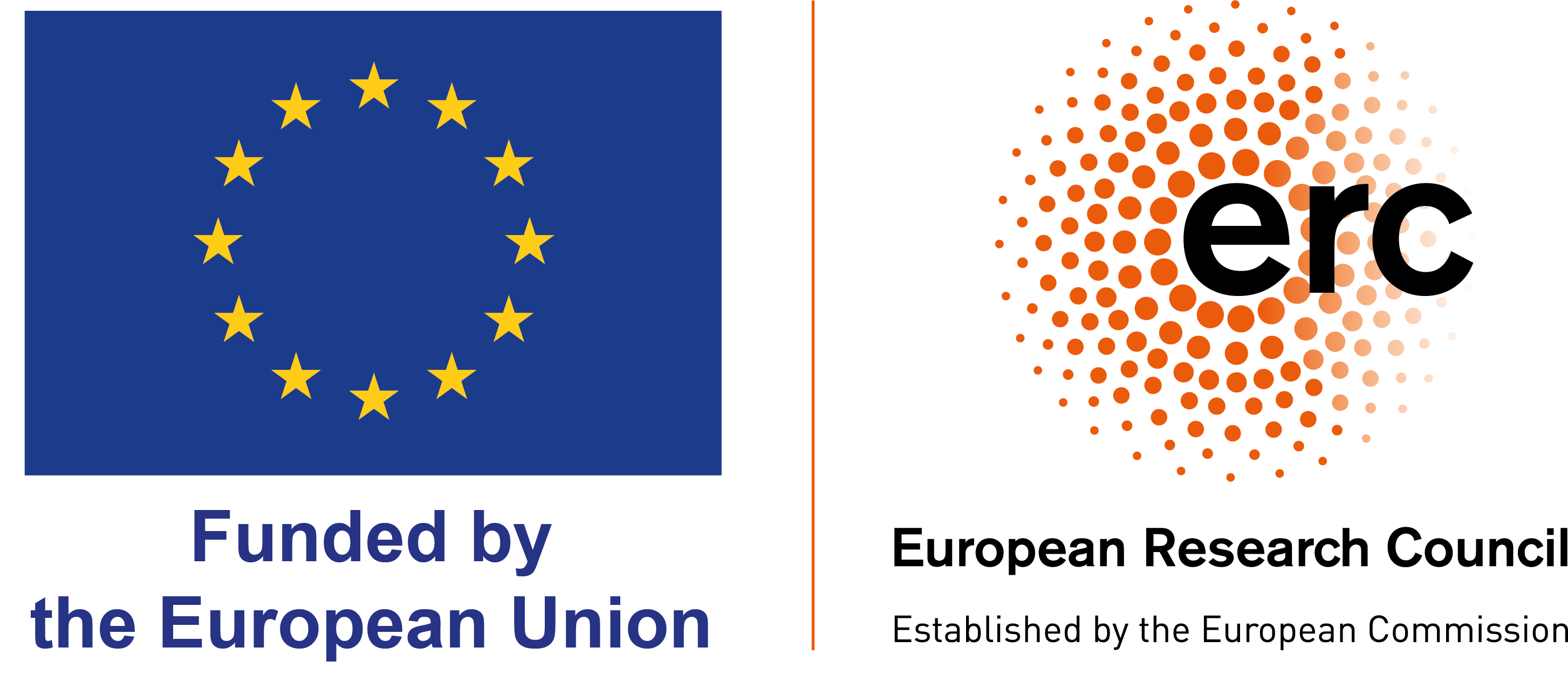}
\end{minipage}%
}
\endgroup
\makeatother

\begin{abstract}
The computation of row echelon form is one of the main bottlenecks in the F4 algorithm.
Several state of the art implementations use a probabilistic algorithm attributed to Monagan, Pearce, and Steel to accelerate this computation. Despite this, no bound on the probability that the algorithm returns an incorrect result appears to be available. In this paper, we provide such a bound. Furthermore, building on this result, we propose a Las-Vegas variant of the F4 algorithm and show experimentally that it can outperform deterministic F4 on some classical examples.
\end{abstract}

\section{Introduction}

The F4 algorithm~\cite{F4} is a widely adopted algorithm for computing \grobner{} bases. It is implemented in many computer algebra systems, such as Giac/Xcas~\cite{giac}, Macaulay2~\cite{M2}, Magma~\cite{magma}, Maple~\cite{maple}, to name a few. The efficiency in the F4 algorithm depends, in particular, on the efficient computation of a row echelon form of a matrix over a finite field.

Efficient linear algebra in the F4 algorithm has been studied before~\cite{gbla,fgla,compact}. 
In this line of work, there exists a probabilistic method that exploits the low-rank structure of the matrices that typically appear in the F4 algorithm. The method has been attributed to Monagan, Pearce, and Steel; we quote~\cite{splitting}:

   {\it ``The probabilistic strategy was first developed by Allan Steel (private communication), and we figured it out independently from studying Magma timings.''}

The method by Monagan, Pearce, and Steel has been used in state of the art implementations, including axcas~\cite{axcas}, Groebner.jl~\cite{demin2024groebnerjl}, and msolve~\cite{msolve}. The algorithm is probabilistic in the Monte-Carlo sense, which means the result could be incorrect. However, to the best of our knowledge, no bound on the probability of an incorrect result is available.
% probabilistic analysis of the algorithm has not been carried out yet.

In this article, 
\begin{itemize}
    \item We provide a practical upper bound on the probability that the algorithm by Monagan, Pearce, and Steel returns an incorrect result (\Cref{sec:ref}). 

    \item We propose a Las-Vegas variant of the F4 algorithm based on the probabilistic algorithm by Monagan, Pearce, and Steel (\Cref{sec:f4}). The construction uses the fact that the output of the probabilistic algorithm is contained in the rowspace of the input matrix, even when the result is incorrect. We also show how, in practice, the probability bounds can be sharpened by combining our analysis with the block structure of Macaulay matrices described by Faug{\`e}re and Lachartre~\cite{fgla}.
\end{itemize}

In practice, these results would allow users to prescribe a lower bound on the probability of correctness and, when needed, obtain a certified result using the Las-Vegas algorithm. Thus, several modern implementations of the F4 algorithm may benefit from these guarantees. Our experiments show that the Las-Vegas approach can provide a speedup over the deterministic F4 algorithm in examples where the cost of certification is small compared with the gain from probabilistic linear algebra.

\paragraph{Acknowledgement.}
We are grateful to Matías Bender, Jean-Guillaume Dumas, Roman Pearce, and Gleb Pogudin for discussions, comments, and ideas.
We thank Fabrice Rouillier for suggesting \Cref{lemma:inclusion} to us.

\section{Probabilistic row echelon form}
\label{sec:ref}

\subsection{Notation}

Let $k$ denote a finite field with $q$ elements.
Let $\cU(k^{m\times n})$ denote the uniform probability distribution on $k^{m\times n}$, that is, the distribution of $m\times n$ matrices with
entries chosen independently and uniformly from $k$.

\subsection{Rank of a random matrix}

\begin{lemma}[{\cite[Lemma 1]{salmond2016}, \cite{orig}}]
\label{lemma:prob-random-square-not-singular}
Let $r \leqslant N$. Let $B \sim \cU(k^{r \times N})$. The probability that $\rank(B) = r$ is
\[
\prod_{i=N-r+1}^{N} \left(1-\frac{1}{q^i}\right).
\]
\end{lemma}
\begin{proof}
    For every $i=0,\ldots,r$, let $E_i$ denote the event that the first $i$ rows of $B$ are linearly independent. In particular, $P(E_0)=1$. It suffices to calculate $P(E_r)$. 
    Note
    \[
    \begin{aligned}
    P(E_r) &= P(E_{r} | E_{r-1}) P(E_{r-1}) + P(E_{r} | \widebar{E_{r-1}}) P(\widebar{E_{r-1}})\\
    &= P(E_{r} | E_{r-1}) P(E_{r-1})\\
    &= P(E_{r} | E_{r-1}) \cdots P(E_{1} | E_{0}) P(E_{0})\\
    &= P(E_{r} | E_{r-1}) \cdots P(E_{1} | E_{0}).
    \end{aligned}
    \]
    Therefore, it suffices to calculate $P(E_i | E_{i-1})$, where $i \in \{1,\ldots,r\}$. If $E_{i-1}$ holds, then the first $i-1$ rows of $B$ span a vector space of dimension $i-1$ with $q^{i-1}$ elements. Hence the $i$-th row of $B$, which is picked uniformly from $k^N$, is outside this span with probability
    \[
    P(E_i | E_{i-1}) = \frac{q^N - q^{i-1}}{q^N} = 1 - \frac{1}{q^{N - i + 1}}.
    \]
    We take the product over $i=1,\ldots,r$, and the claim follows.
\end{proof}

\begin{corollary}
\label{lemma:prob-random-product-rpsn-nr}
Let $r \leqslant N$ and $s \geqslant 1$. Let $B \sim \cU(k^{(r+s-1) \times N})$. 
Let $A \in k^{N\times r}$ be arbitrary but fixed with $\rank(A) = r$.
The probability that $\rank(B A) = r$ is
\[
\prod_{i=s}^{r+s-1} \left(1 - \frac{1}{q^i}\right).
\]
\end{corollary}
\begin{proof}
Since $\rank(A) = r$, the map $b\mapsto bA$ is surjective, and its kernel is a linear subspace of dimension $N-r$. Therefore, for every $y\in k^r$, exactly $q^{N-r}$ elements of $k^N$ map to $y$.
Hence, if $b$ is uniformly distributed in $k^N$, then, for every $y\in k^r$,
\[
P(b A = y) = \frac{q^{N-r}}{q^N}=\frac{1}{q^r}.
\]
Thus, if $b$ is uniformly distributed in $k^N$, then $bA$ is uniformly distributed in $k^r$.
Since the rows of $B$ are independent, we have
\[
BA\sim\cU(k^{(r+s-1)\times r}).
\]
We apply \Cref{lemma:prob-random-square-not-singular} to $(BA)^T$ and the claim follows.
\end{proof}

\begin{corollary}
\label{lemma:prob-random-product-rpsn-nm}
Let $r \leqslant N$ and $s \geqslant 1$. Let $B \sim \cU(k^{(r+s-1) \times N})$. 
Let $A \in k^{N\times M}$ be arbitrary but fixed with $\rank(A) = r$. 
The probability that $\rank(B A) = r$ is
\[
\prod_{i=s}^{r+s-1} \left(1 - \frac{1}{q^i}\right).
\]
\end{corollary}
\begin{proof}
Let $A = GH$ be a rank factorization of $A$, that is, $G \in k^{N\times r}, H \in k^{r \times M}$ with $\rank(G) = \rank(H) = r$. Then,
\[
BA = BGH.
\]
Note $\rank(BGH) = \rank(BG)$, since right multiplication by a matrix with linearly independent rows does not change the rank. Therefore, it suffices to estimate the probability that $\rank(BG) = r$. We apply~\Cref{lemma:prob-random-product-rpsn-nr} to $B, G$ and the claim follows.
\end{proof}

Let $r \leqslant N$ and $s \geqslant 1$. Let $B \sim \cU(k^{(r+s-1) \times N})$. 
Let
\[
\label{eq:concat}
A \coloneqq \begin{bmatrix}
    A_1 & A_2\\
    A_3 & A_4
\end{bmatrix} \in k^{(R+N)\times (R+M)}
\quad
\text{and}
\quad
A^{\ast} \coloneqq \begin{bmatrix}
    A_1 & A_2\\
    B A_3 & B A_4
\end{bmatrix} \in k^{(R+r+s-1)\times (R+M)},
\]
where $A_1 \in k^{R\times R}, A_2 \in k^{R\times M}, A_3 \in k^{N\times R}, A_4 \in k^{N\times M}$, with $\rank(A_1) = R$ and $\rank(A) = R + r$.

\begin{corollary}\label{lemma:schur}
The probability that $\rank(A) = \rank(A^\ast)$ is
\[
\prod_{i=s}^{r+s-1} \left(1 - \frac{1}{q^i}\right).
\]
\end{corollary}
\begin{proof}
Let $C \coloneqq A_4 - A_3 A_1^{-1}A_2 \in k^{N \times M}$. Consider the Schur complement of $A^\ast$:
\[
A^\ast = 
\begin{bmatrix}
    I_{R\times R} & 0\\
    B A_3 A_1^{-1} & I_{(r+s-1) \times (r+s-1)}
\end{bmatrix}
\begin{bmatrix}
    A_1 & A_2\\
    0 & B C
\end{bmatrix}.
\]
Note $\rank(A^\ast) = \rank(A_1) + \rank(B C)$. Therefore, it suffices to estimate the probability that
\[
\rank(B C) = r.
\]
Note $\rank(A) = \rank(A_1) + \rank(C)$. This can be seen by considering the Schur complement of $A$. Therefore, $\rank(C) = r$. We apply~\Cref{lemma:prob-random-product-rpsn-nm} to $B, C$ and the claim follows.
\end{proof}

\subsection{Probabilistic row echelon form}
\label{sec:prob-echelon}

Let $A \in k^{(R+N)\times (R + M)}$ be of the form
\begin{equation}
    \label{eq:concat-2}
A := \begin{bmatrix}
    A_1\\
    A_2
\end{bmatrix},
\end{equation}
where $A_1 \in k^{R\times (R+M)}, A_2 \in k^{N\times (R+M)}$, such that $A_1$ is in row echelon form with $\rank(A_1) = R$, and $\rank(A) = R + r$. 

Our goal is to compute a row echelon form of this matrix.
We first give an informal description of the algorithm following~\cite[Section 4.2]{msolve}. Let $L := A_1$, which is already in row echelon form. Partition the rows of $A_2$ into blocks of a given size, say $\ell$ consecutive rows form a block. More precisely, let $T := \lceil N/\ell \rceil$ and
\[
\begin{array}{cc}
A_2 := \begin{bmatrix}
    A_{2,1}\\
    \vdots\\
    A_{2,T}
\end{bmatrix},&
\quad\begin{aligned}
&A_{2,1},\ldots,A_{2,T-1} \in k^{\ell\times (R+M)},\\ &A_{2,T} \in k^{\min(\ell, N - (T-1) \ell)\times (R+M)}. 
\end{aligned}
\end{array}
\]
We take a random linear combination of the rows in a block and reduce it with respect to $L$. If the outcome is non-zero, we have found a new pivot row and add it to $L$. We then take another random linear combination and repeat the procedure. We stop with the current block once we have either reduced $\ell$ linear combinations (or, perhaps, fewer for $A_{2,T}$) or once a reduction to zero happens. The probability of getting zero by chance is small, and can be decreased by doing more than one reduction to zero before stopping with the block.

\begin{algorithm}[H]
\caption{Probabilistic row echelon form (Monagan, Pearce, and Steel, cf.~\cite{splitting})}
\label{alg:randomized-rowspace-first}
\vspace{-0.4em}
\begin{description}[itemsep=0pt]
\item[Input:] $A \in k^{(R+N) \times (R+M)}$ as in~\eqref{eq:concat-2}; an integer $1 \leqslant \ell \leqslant N$, the size of the block; an integer $s \geqslant 1$, the number of reductions to zero before stopping with the current block.
\item[Output:] a row echelon form of $A$.
\end{description}

\begin{enumerate}[label = \textbf{(Step~\arabic*)}, leftmargin=*, align=left, labelsep=4pt, itemsep=2pt, topsep=2pt]
    \item\label{step:first} Set $L \coloneqq A_1$ and $T \coloneqq \lceil N/\ell \rceil$.
    \item\label{step:randomized-rank:loop-3} For $t=1,\ldots,T$ do
    \begin{enumerate}[labelsep=4pt, itemsep=4pt,topsep=4pt, label = (\alph*), ref = \theenumi(\alph*)]
    % \item Set $I_t \coloneqq \{ 1+(t-1)T,\ldots,\min(t T, N) \}$.
    \item Set $z := 0$ and $\ell_t := \min(\ell, N - (t-1)\ell)$.
    \item\label{step:randomized-rank:loop-2} For $i=1,\ldots,\ell_t+s-1$ do
    \begin{enumerate}[labelsep=4pt, itemsep=4pt, topsep=4pt, label = \roman*., ref = \theenumii\roman*]
        \item \label{step:random-vec} Let $b_i \in k^{\ell_t}$ be chosen uniformly.
        \item Compute $c \coloneqq b_i^T A_{2,t}$. \algorithmiccomment{{A random linear combination of rows in the $t$-th block.}}
        \item Compute $c^\ast \in k^{R+M}$, the reduction of $c$ with respect to the row echelon form $L$.
        \item If \(c^*=0\), then set $z\coloneqq z+1$.
        If $z=s$, then \textbf{break} out of the loop of \ref{step:randomized-rank:loop-2}.
        \item If $c^*\neq0$, then append $c^\ast$ to the rows of $L$, maintaining row echelon form.
    \end{enumerate}
    \end{enumerate}
    \item {\bf Return} $L$.
\end{enumerate}
\end{algorithm}

\begin{remark}
\label{remark:r-is-zero}
    In the case $R = 0$, \Cref{alg:randomized-rowspace-first} yields a probabilistic algorithm for computing a row echelon form of an arbitrary matrix $A \in k^{N\times M}$.
\end{remark}

\begin{remark}[Variants of the algorithm]
\Cref{alg:randomized-rowspace-first} stops processing a block after $s$ reductions to zero in total. The strategy described in~\cite{splitting}
instead requires $s$ consecutive reductions to zero. The two methods coincide when $s=1$. We leave the
probabilistic analysis of the latter variant for future work.
\end{remark}

\begin{proposition}
\label{prop:result}
Let $T \coloneqq \lceil N / \ell \rceil$. \Cref{alg:randomized-rowspace-first} returns a correct result with probability at least
\[
\prod_{i=s}^{\ell+s-1}\left(1-\frac{1}{q^i}\right)^T.
\]
Furthermore, upon termination, we have $\operatorname{rowspace}(L) \subseteq \operatorname{rowspace}(A)$ (even when the result is incorrect).
\end{proposition}

\begin{proof}[Proof of~\Cref{prop:result}]
    It suffices to estimate the probability that, upon termination, we have
    \begin{center}
    $\operatorname{rowspace}(L) \subseteq \operatorname{rowspace}(A)$\quad and\quad $\rank(L) = \rank(A)$.    
    \end{center}
    Initially, at \ref{step:first}, we have $\operatorname{rowspace}(L) \subseteq \operatorname{rowspace}(A)$. 
    At every iteration in the loop at~\ref{step:randomized-rank:loop-2}, we have $c^\ast \in \operatorname{rowspace}(A)$, and, hence, upon termination, $\operatorname{rowspace}(L) \subseteq \operatorname{rowspace}(A)$. 

    Therefore, it suffices to estimate the probability that
    \[
    \rank(L) = \rank(A).
    \]
    For every $t=0,\ldots,T$, let $A^{[t]} \in k^{(R+\min(t \ell, N))\times (R+M)}$ denote the matrix that consists of the rows of $A_1$ together with the first $\min(t \ell, N)$ rows of $A_2$. Let $L^{[t]}$ denote the value of $L$ after the $t$-th block has been
    processed, that is, after the $t$-th iteration in \ref{step:randomized-rank:loop-3}. Let $L^{[0]}:=A_1$. 
    Let $E_t$ denote the event that $\rank(L^{[t]}) = \rank(A^{[t]})$. We wish to estimate the probability that $E_T$ holds. Note $E_{0}$ holds.
    Then,
    \[
    \begin{aligned}
    P(E_T) &= P(E_T | E_{T-1})P(E_{T-1}) + P(E_T | \widebar{E_{T-1}})P(\widebar{E_{T-1}})\\
    &\geqslant P(E_T | E_{T-1})P(E_{T-1}) \\
    &\geqslant P(E_T | E_{T-1}) \cdots P(E_1 | E_{0}).
    \end{aligned}
    \]
    Therefore, it suffices for every $t=1,\ldots,T$ to estimate $P(E_t | E_{t-1})$. Let \[r_t \coloneqq \rank(A^{[t]}) - \rank(A^{[t-1]}).\]
    Let $t \in \{1,\ldots,T\}$ be an arbitrary iteration of the loop in \ref{step:randomized-rank:loop-3}.
    Let $B_t \in k^{(r_t+s-1)\times\ell_t}$ be the random matrix with the rows $b_1,\ldots,b_{r_t+s-1}$. 
    Assume $E_{t-1}$ holds. Then $\operatorname{rowspace}(L^{[t-1]})= \operatorname{rowspace}(A^{[t-1]})$, and $E_t$ holds if and only if
    \[
    \rank
    \begin{bmatrix}
    L^{[t-1]}\\
    A_{2,t}
    \end{bmatrix}
    =
    \rank
    \begin{bmatrix}
    L^{[t-1]}\\
    B_tA_{2,t}
    \end{bmatrix}.
    \]
    We apply~\Cref{lemma:schur} to the above equality and obtain
    \[
    P(E_t | E_{t-1}) = \prod_{i=s}^{r_t + s - 1}\left(1-\frac{1}{q^i}\right).
    \]
    Note for every $t=1,\ldots,T$ we have $r_t \leqslant \ell$, and, therefore,
    \[
    P(E_t | E_{t-1}) \geqslant \prod_{i=s}^{\ell + s - 1}\left(1-\frac{1}{q^i}\right).
    \]
    We take the product over $t=1,\ldots,T$, and the claim follows.
\end{proof}

\begin{remark}\label{remark:sharp}
    The bound in \Cref{prop:result} is sharp if and only if the matrix is of full row rank, that is, $r = N$, and $\ell$ divides $N$. Indeed, in the notation of the proof of~\Cref{prop:result}, these conditions are equivalent to the fact that for every $t = 1,\ldots,T$ we have
    $r_t = \ell$
        and
    $P(E_t | \widebar{E_{t - 1}}) = 0$, which make the inequalities sharp.
\end{remark}

\begin{remark}
If the rank $r$ is known, then the choice $\ell \coloneqq N$ combined with a direct application of \Cref{lemma:schur} yields an exact probability instead of a lower bound as in \Cref{prop:result}. This may be useful when the rank is known but a row echelon form is required.
\end{remark}

\subsection{Numerical examples}

In~\Cref{tab:remark-1}, we list upper bounds on the probability of failure in~\Cref{alg:randomized-rowspace-first} provided by~\Cref{prop:result}. The examples come from classical benchmarks for the F4 algorithm~\cite{F4}. For every example, we consider only one matrix, the one with the largest number of rows $N$ (in the F4 algorithm this corresponds to the matrix with the largest number of polynomials reduced at once). The block size is $\ell \coloneqq\lceil\sqrt{N}\rceil$, a heuristic often used in practice. We consider different values of $q$. Probability bounds in the table are rounded up. We observe that incrementing $s$ by one decreases the probability bound roughly by a factor of $q$.

\begin{table}[H]
\renewcommand{\arraystretch}{1.2} % Default value: 1
\centering
\begin{tabular}{l||c|c|r|r||r|r}
Example & $N$ & $\ell$ & $q$ & $1/q$& \thead{\Cref{prop:result}\\with $s\coloneqq1$} & \thead{\Cref{prop:result}\\with $s\coloneqq2$}\\
\hline\hline
\multirow{4}{*}{Cyclic-8} & \multirow{4}{*}{1,216} & \multirow{4}{*}{35} & $2^8 - 5$ & $4 \cdot 10^{-3}$ & $2 \cdot 10^{-1}$ & $6 \cdot 10^{-4}$\\
&&& $2^{16}-15$ & $2 \cdot 10^{-5}$ & $6\cdot10^{-4}$ & $9 \cdot 10^{-9}$\\
&&& $2^{32}-5$ & $3 \cdot 10^{-10}$ & $9\cdot10^{-9}$ & $2 \cdot 10^{-18}$\\
&&& $2^{64}-59$ & $6 \cdot 10^{-20}$ & $2 \cdot 10^{-18}$ & $2 \cdot 10^{-37}$\\
\hline
\multirow{4}{*}{Cyclic-9} & \multirow{4}{*}{5,885} & \multirow{4}{*}{77} & $2^8-5$ & $4\cdot10^{-3}$ & $3 \cdot 10^{-1}$ & $2 \cdot 10^{-3}$\\
&&& $2^{16}-15$ & $2 \cdot 10^{-5}$ & $2\cdot10^{-3}$ & $2 \cdot 10^{-8}$\\
&&& $2^{32}-5$ & $3 \cdot 10^{-10}$ & $2\cdot10^{-8}$ & $5 \cdot 10^{-18}$\\
&&& $2^{64}-59$ & $6 \cdot 10^{-20}$ & $5 \cdot 10^{-18}$ & $3 \cdot 10^{-37}$\\
\end{tabular}
\caption{Upper bounds on the probability of failure in~\Cref{alg:randomized-rowspace-first} for the largest matrix that appears in the F4 algorithm. 
% The fraction $1/q$ is rounded up.
}
\label{tab:remark-1}
\end{table}

\subsection{Choosing parameters in probabilistic row echelon form}

\Cref{prop:result} may be inconvenient to use, because the asymptotic behavior of the bound is not readily apparent.
We provide potentially a more convenient bound based on~\Cref{prop:result}, which could guide the choice of the parameters in the algorithm.
\begin{proposition}
\label{prop:result-3}
\Cref{alg:randomized-rowspace-first} returns an incorrect result with probability at most
\[
\frac{\left\lceil N/\ell \right\rceil}{q^{s-1}(q-1)}.
\]
\end{proposition}
\begin{proof}
Let $T \coloneqq \lceil N / \ell \rceil$. By combining~\Cref{prop:result} with the Bernoulli's inequality, we have
\[
\begin{aligned}
    \prod_{i=s}^{\ell+s-1}\left(1-\frac{1}{q^i}\right)^T &\geqslant \left(1 - \frac{1}{q^s} - \ldots - \frac{1}{q^{\ell+s-1}}\right)^T\\
    &= \left(1 - \frac{1}{q^s}\cdot\frac{1-\frac{1}{q^\ell}}{1 - \frac{1}{q}}\right)^T\\
    &\geqslant \left(1 - \frac{1}{q^s}\cdot\frac{q}{q - 1}\right)^T\\
    &\geqslant 1 - \frac{1}{q^s}\cdot\frac{q}{q - 1}\cdot T\\
    &= 1 - \frac{T}{q^{s-1}(q-1)}.&&\qedhere
\end{aligned}
\]
\end{proof}

% \begin{remark}[\Cref{prop:result} vs.~\Cref{prop:result-3}]
% We observed that the bounds provided by~\Cref{prop:result} and~\Cref{prop:result-3} are typically close. For example, for the examples in~\Cref{tab:remark-1}, the bounds agree to machine double precision.
% \end{remark}

In practice, the user may want to specify an upper bound on the probability of incorrect result, a real number $\varepsilon$, with $0 < \varepsilon < 1$. Using \Cref{prop:result-3}, for fixed $N,\ell,q$, we can compute an integer $s_\varepsilon$, such that the probability of failure in~\Cref{alg:randomized-rowspace-first} applied with $s := s_\varepsilon$ does not exceed $\varepsilon$:

\begin{proposition}\label{prop:compute:s}
Let $A$ be as in~\eqref{eq:concat-2}. Let $1 \leqslant \ell \leqslant N$ and $0 < \varepsilon < 1$ be fixed. The probability that \Cref{alg:randomized-rowspace-first} applied to $A$, $\ell$, and $s_\varepsilon$ returns an incorrect result is at most $\varepsilon$, where
\begin{equation}
% \label{eq:compute-s}
s_\varepsilon \coloneqq \left\lceil \frac{\log_2 \left(\frac{\left\lceil N/\ell \right\rceil}{\varepsilon(q-1)}\right)}{\log_2q}+1\right\rceil.
\end{equation}
\end{proposition}
\begin{proof}
The claim follows readily from~\Cref{prop:result-3}.
\end{proof}

\section{Probabilistic row echelon form in the F4 algorithm}\label{sec:f4}

The goal of this section is to introduce a Las-Vegas variant of the F4 algorithm (\Cref{alg:lv-f4}) using probabilistic row echelon form. Let $k$ be a finite field with $q$ elements.

\subsection{Algebraic background}

We recall some standard definitions, and also refer to \cite[Chapter 2]{cox}. Let $k[x_1,\ldots,x_\tau]$ denote the ring of polynomials in the indeterminates $x_1,\ldots,x_\tau$ over $k$. Let $\mathbb{N}_0 \coloneqq \{0,1,2,\ldots\}$. A monomial is a product $x_1^{e_1} \cdots x^{e_\tau}_\tau$ with $e_1,\ldots,e_\tau \in \mathbb{N}_0$. 
A term is a constant multiple of a monomial. 
Let $\mathcal{M}$ denote the set of monomials in $k[x_1,\ldots,x_\tau]$. A monomial ordering $\prec$ is a total ordering on $\mathcal{M}$ such that $1 \prec u$ for $u\neq1$ and $u \prec v$ implies $u w \prec v w$ for every $u,v,w\in\mathcal{M}$. 

Let $f \in k[x_1,\ldots,x_\tau]$. 
If $f \neq 0$, then let $\operatorname{lm}(f)$ and $\operatorname{lt}(f)$ 
denote the leading monomial and the leading term 
of $f$ with respect to monomial ordering $\prec$, respectively.
Let $m \in \mathcal{M}$, and let $\operatorname{coeff}(f, m) \in k$ denote the coefficient in front of the monomial $m$ in $f$. Let $\operatorname{supp}(f) \subset \mathcal{M}$ denote the set of monomials that enter in $f$ with nonzero coefficients.

Let $G \subset k[x_1,\ldots,x_\tau]$ be a finite subset. If $f\neq0$, we say that $f$ is reducible with respect to $G$ if there exist $g \in G \setminus \{0\}$, $u \in \mathcal{M}$, and $c \in k\setminus \{0\}$, such that $\operatorname{lt}(c\cdot u\cdot g) = \operatorname{lt}(f)$. Then, $c \cdot u \cdot g$ is called a reducer of $f$, and $f - c \cdot u \cdot g$ is the result of a single reduction of $f$ with respect to $G$. We say that $f$ reduces to $f^*$ with respect to $G$ if there exists a chain of reductions so that $f^* = f - c_1 u_1 g_1 - \ldots - c_\nu u_\nu g_\nu$ with $g_1,\ldots,g_\nu \in G \setminus\{0\}$, $u_1,\ldots,u_\nu \in \mathcal{M}$, and $c_1,\ldots,c_\nu \in k\setminus \{0\}$. 

Let $f, g \in k[x_1,\ldots,x_\tau]\setminus \{0\}$. Let 
$\operatorname{mult}(f, g) \coloneqq \operatorname{lcm}(\operatorname{lm}(f), \operatorname{lm}(g)) / \operatorname{lt}(f)$. The S-polynomial of $f$ and $g$ is defined as
\[
\operatorname{Spoly}(f, g) \coloneqq  f\cdot \operatorname{mult}(f, g) - g \cdot\operatorname{mult}(g, f).
\]

Let $I = \langle G \rangle$ denote the ideal generated by a finite set of polynomials $G \subset k[x_1,\ldots,x_\tau]$. 
A set $G$ is a \grobner{} basis of $I$ with respect to the monomial ordering $\prec$ if for every $f \in I \setminus \{0\}$ there exists $g \in G$ such that $\operatorname{lm}(g)$ divides $\operatorname{lm}(f)$.
Every ideal in $k[x_1,\ldots,x_\tau]$ has a \grobner{} basis~\cite[Chapter 2, \S 5]{cox}.

\subsection{Macaulay matrices}

Fix a monomial ordering $\prec$.

\begin{definition}[Macaulay matrix]
\label{def:mac}
Let $P,G\subset k[x_1,\ldots,x_\tau]$. We construct $U \subset k[x_1,\ldots,x_\tau]$ in the following way. Starting with $U\coloneqq P$, for every monomial $m$ in $\operatorname{supp}(u)$ for some $u \in U$, we choose a single $g \in G \setminus \{0\}$ such that $m$ is divisible by $\lm(g)$ (if such exists) and add $\frac{m}{\lm(g)}g$ to $U$. We recursively repeat the procedure for the monomials that have not yet been processed in the supports of the added polynomials.

The Macaulay matrix $\Mac(P,G) \in k^{\mu\times\nu}$ is the coefficient matrix of $U = \{u_1,\ldots,u_\mu\}$ with respect to the monomials $m_1 \succ \ldots \succ m_\nu$ occurring in the supports of $U$, that is, $\operatorname{Mac}(P,G)_{i,j} = \operatorname{coeff}(u_i, m_j)$ for every $i = 1,\ldots,\mu$ and $j = 1,\ldots,\nu$.
\end{definition}

\begin{notation}
    We may implicitly convert between polynomials and the rows of $\operatorname{Mac}(P,G)$. In particular, we may treat $h \in \operatorname{rowspace}(\operatorname{Mac}(P,G))$ both as an element of $k^\nu$ and as an element of $k[x_1,\ldots,x_\tau]$.
\end{notation}

\begin{notation}
    Let $P, G \subset k[x_1,\ldots,x_\tau]$. Let $M^*$ be a row echelon form of $M\coloneqq\Mac(P,G)$, and let $F^*\subset k[x_1,\ldots,x_\tau]$ be polynomials corresponding to the rows of $M^*$. Denote
    \[
    \operatorname{Rem}_G(M^*) \coloneqq
\{h\in F^* \setminus \{0\} \mid
\nexists g\in G\setminus\{0\}:
\operatorname{lm}(g)\mid\operatorname{lm}(h)\}.
    \]
\end{notation}

\begin{lemma}[Polynomial reduction via Macaulay matrix]
\label{lemma:mac-ref}
If $h\in\operatorname{rowspace}(M^*)$ is nonzero and not reducible with respect to $G$, then there exists $h^*\in\operatorname{Rem}_G(M^*)$
such that $\operatorname{lm}(h)=\operatorname{lm}(h^*)$.
\end{lemma}
\begin{proof}
Assume $h \neq 0$. Note $h$ is a $k$-linear combination of the elements of $F^*$. Since $M^*$ is in row echelon form, the leading monomials of the
nonzero elements of $F^*$ are distinct. Therefore, the leading monomial of a linear combination is the leading monomial of one of the summands. Hence, the leading monomial of $h$ is the leading
monomial of some $h^*\in F^*$. Since $h$ is not reducible with respect to $G$, neither is $h^*$. Hence $h^*\in\operatorname{Rem}_G(M^*)$.
\end{proof}

We recall an algorithm for testing if a set of polynomials is a \grobner{} basis using Macaulay matrices.

\begin{proposition}[\grobner{} basis test]\label{prop:test-gb}
Let $G\subset k[x_1,\ldots,x_\tau]$. Let
$
P \coloneqq
\{\operatorname{Spoly}(f,g)\mid f,g\in G\setminus\{0\}\},
$
and let \(M^*\) be a row echelon form of $M\coloneqq\Mac(P,G)$. Then
\[
G\text{ is a \grobner{} basis of $\langle G \rangle$}
\quad\Longleftrightarrow\quad
\operatorname{Rem}_G(M^*)=\varnothing.
\]
\end{proposition}
\begin{proof}
Assume $\operatorname{Rem}_G(M^*)=\varnothing$. Suppose $G$ is not a \grobner{} basis. By Buchberger's criterion~\cite[Ch.~2, \S 6]{cox}, there exist
$f,g\in G\setminus\{0\}$ such that
$s\coloneqq\operatorname{Spoly}(f,g)$ reduces to $s^*$ with respect to $G$, and $s^*$ is nonzero and not reducible with respect to $G$. By the construction of $\Mac(P,G)$, we may reduce $s$ using, for each reducible leading monomial, the reducer selected in~\Cref{def:mac}. Hence the result $s^*$ of this sequence of reductions belongs to $\operatorname{rowspace}(M) =
\operatorname{rowspace}(M^*)$.
By~\Cref{lemma:mac-ref}, there exists
$h^*\in\operatorname{Rem}_G(M^*)$, which contradicts
$\operatorname{Rem}_G(M^*)=\varnothing$.

Conversely, assume $G$ is a \grobner{} basis. Every row of $M$ belongs to $\langle G\rangle$, and therefore so does every row of $M^*$. Hence, for every nonzero row $h^*$ of $M^*$, there exists
$g\in G\setminus\{0\}$ such that
$\lm(g)\mid\lm(h^*),$ whence $\operatorname{Rem}_G(M^*)=\varnothing$.
\end{proof}

\subsection{Monte-Carlo variant of the F4 algorithm}

We recall a simplified Monte-Carlo variant of the F4 algorithm resembling the one from~\cite{msolve,splitting}.

\begin{algorithm}[H]
\caption{Monte-Carlo variant of the F4 algorithm (simplified version, cf.~{\cite[Algorithm F4]{F4}})}
\label{alg:mc-f4}
\vspace{-0.4em}
\begin{description}[itemsep=0pt]
\item[Input:] $F \subset k[x_1,\ldots,x_\tau]$ and a real number $0 < \varepsilon < 1$.
\item[Output:] $G \subset k[x_1,\ldots,x_\tau]$, a \grobner{} basis of $\langle F \rangle$.
\end{description}

\begin{enumerate}[label = \textbf{(Step~\arabic*)}, leftmargin=*, align=left, labelsep=4pt, itemsep=2pt, topsep=2pt]
    \item Let $\varepsilon_1,\varepsilon_2,\ldots$ be real numbers such that $0 < \varepsilon_i < 1$ and $\varepsilon_1 + \varepsilon_2 + \ldots \leqslant \varepsilon$.
    \item Let $G \coloneqq F$.
    \item \label{step:alg1-3} For every $i = 1, 2, 3\ldots$ {\bf do}
    \begin{enumerate}[label = (\alph*), ref = \theenumi (\alph*), leftmargin=*, align=left, labelsep=2pt, itemsep=2pt]
    \item Let $P \coloneqq \{\operatorname{Spoly}(f,g) \mid f, g \in G \setminus \{0\}\}$. 
    \item Let $M := \mathrm{Mac}(P, G) \in k^{\mu \times \nu}$. 
    \item \label{step:compute-s} Let $1 \leqslant \ell \leqslant \mu$. Let $s_{\varepsilon_i}$ be as in~\Cref{prop:compute:s} applied with $N := \mu$.
    \item \label{step:alg1-3-3} Compute $M^*$, a row echelon form of $M$, by applying~\Cref{alg:randomized-rowspace-first} to $M$, $\ell$ and $s_{\varepsilon_i}$ (\Cref{remark:r-is-zero}).
    \item If $\operatorname{Rem}_G(M^*) = \varnothing$, then {\bf return} $G$; otherwise, set $G \coloneqq G \cup \operatorname{Rem}_G(M^*)$ and continue.
    \end{enumerate}
\end{enumerate}
\end{algorithm}

We first note that even when the output of this algorithm is not a correct \grobner{} basis, it still has the following useful property:

\begin{lemma}\label{lemma:inclusion}
Let $F \subset k[x_1,\ldots,x_\tau]$ and $0 < \varepsilon < 1$. Let $G$ be the output of the Monte-Carlo variant of the F4 algorithm (\Cref{alg:mc-f4}) applied to $F$ and $\varepsilon$. Then, $\langle G \rangle = \langle F \rangle$.
\end{lemma}
\begin{proof}
At every iteration in \ref{step:alg1-3}, the rows of $M=\Mac(P,G)$ are either monomial multiples or S-polynomials of the elements of $G$, and thus belong to $\langle G \rangle$.
By~\Cref{prop:result}, every row of $M^*$ also belongs to $\langle G\rangle$. Hence, $\operatorname{Rem}_G(M^*)\subseteq\langle G\rangle$, so $\langle G \rangle = \langle F \rangle$ holds at every iteration.
\end{proof}

\begin{example}\label{example:incosistent}
Consider the case when \Cref{alg:mc-f4} returns $G$ with $1 \in G$. Then \Cref{lemma:inclusion} implies $1 \in \langle F \rangle$.
\end{example}

\begin{proposition}[{cf. \cite[Theorem 2.2]{F4}}]
\label{theorem:mc-f4}
    \Cref{alg:mc-f4} terminates. \Cref{alg:mc-f4} returns a correct result with probability at least $1-\varepsilon$.
\end{proposition}

\begin{proof}
We first prove termination. Let $i \in \mathbb{N}$. Let $G_i, M_i$, and $M_i^*$ denote the values of $G, M, M^*$, respectively, at the $i$-th iteration of \ref{step:alg1-3}. Let also
\[
J_i\coloneqq
\langle \operatorname{lm}(g)\mid g\in G_i\setminus\{0\}\rangle.
\]
If the algorithm does not terminate at iteration $i$, then
$\operatorname{Rem}_{G_i}(M_i^*)\neq\varnothing$. Hence, for every
$h\in\operatorname{Rem}_{G_i}(M_i^*)$,
$\operatorname{lm} (h)\notin J_i,$
and therefore
$J_i\subsetneq J_{i+1}.$
Assume the algorithm does not terminate. Then, $J_1 \subsetneq J_2 \subsetneq \ldots$ is an ascending chain of ideals, which contradicts $k[x_1,\ldots,x_\tau]$ being Noetherian.

We now prove correctness. Suppose that  every application of~\Cref{alg:randomized-rowspace-first} in \ref{step:alg1-3-3} returns
a correct row echelon form.
Then the algorithm terminates, and at the final iteration
$\operatorname{Rem}_G(M^*)=\varnothing$. Then $G$ is a \grobner{} basis of $\langle G\rangle$ by \Cref{prop:test-gb}. By \Cref{lemma:inclusion}, $\langle G\rangle=\langle F\rangle$, hence $G$ is a \grobner{} basis of $\langle F\rangle$.

Consider now the probability. At the iteration $i$, the probability that \Cref{alg:randomized-rowspace-first} in \ref{step:alg1-3-3} returns an incorrect result is at most $\varepsilon_i$, so the overall probability of an incorrect result is at most $\varepsilon_1 + \varepsilon_2 + \ldots \leqslant \varepsilon$. 
\end{proof}

\begin{remark}[Block structure of Macaulay matrices]
    Following Faug{\`e}re and Lachartre~\cite[Section 3]{fgla}, under a suitable permutation of rows and columns, the Macaulay matrix $\operatorname{Mac}(P,G)$ decomposes into blocks, $A_1 \in k^{R \times(R+M)}$ and $A_2 \in k^{N \times (R+M)}$, where $A_1$ is already in row echelon form. These blocks match the structure of~\eqref{eq:concat-2}. In particular, $N \leqslant |P|$, where $|P|$ denotes the number of elements in $P$.
    
    In the context of~\Cref{alg:mc-f4}, we may use this decomposition before computing row echelon form via \Cref{alg:randomized-rowspace-first} in \ref{step:alg1-3-3}. It is then sufficient to compute $s_{\varepsilon_i}$ in \ref{step:compute-s} by applying \Cref{prop:compute:s} with $N:=|P|$, instead of $N:=\mu$. Thus, the probability bound depends on the number of lower rows (S-polynomials) rather than on the total number of rows of the Macaulay matrix.
\end{remark}

\begin{example}[$|P|$ vs. $\mu$]
\label{remark:p-vs-u}
    In the F4 algorithm, $|P|$ is the number of S-polynomials, and $\mu$ is the total number of rows in the Macaulay matrix. Note $|P|$ could be much smaller than $\mu$.
    In the univariate case, with $G = \{x^d + 1, x + 1\}$, at the first iteration of F4, there is a single S-polynomial, so $|P| = 1$, but $\mu = d$. 
    In the multivariate case, consider $G = \{ (x_1 + \ldots + x_\tau)^d + 1, x_1 \cdots x_\tau + 1\}$ with $d \leqslant \tau$. In degrevlex monomial ordering, for $k$ of characteristic zero or sufficiently large, we have $|P| = 1$ but $\mu = \binom{\tau+d-2}{d-1}$.
\end{example}

\subsection{Las-Vegas variant of the F4 algorithm}\label{sec:f4-lv}

Based on the observation made in \Cref{lemma:inclusion}, we devise the following algorithm.

\begin{algorithm}[H]
\caption{Las-Vegas variant of the F4 algorithm}
\label{alg:lv-f4}
\vspace{-0.4em}
\begin{description}[itemsep=0pt]
\item[Input:] $F \subset k[x_1,\ldots,x_\tau]$.
\item[Output:] $G \subset k[x_1,\ldots,x_\tau]$, a \grobner{} basis of $\langle F \rangle$.
\end{description}

\begin{enumerate}[label = \textbf{(Step~\arabic*)}, leftmargin=*, align=left, labelsep=4pt, itemsep=2pt, topsep=2pt]
    \item Let $0<\varepsilon<1$ be an arbitrary real number.
    \item \label{alg:step:mc} Let $G$ be the output of the Monte-Carlo variant of the F4 algorithm (\Cref{alg:mc-f4}) applied to $F$ and $\varepsilon$.
    \item \label{alg:step:test} Check if $G$ is a \grobner{} basis of $\langle G \rangle$, for example, by testing that all S-polynomials reduce to zero via~\Cref{prop:test-gb}. If it is, then {\bf return} $G$. Otherwise, {\bf go to} \ref{alg:step:mc}.
\end{enumerate}
\end{algorithm}

\begin{proposition}
\label{theorem:lv-f4}
    \Cref{alg:lv-f4} terminates with probability one. The output of \Cref{alg:lv-f4} is correct.
\end{proposition}

\begin{proof}
The algorithm does not terminate when \ref{alg:step:mc} produces an incorrect result indefinitely. The probability of that is at most $\varepsilon \cdot \varepsilon \cdot \varepsilon \cdots = 0$ by \Cref{theorem:mc-f4}.

Correctness follows from \Cref{lemma:inclusion} and \Cref{prop:test-gb}.
\end{proof}

\subsection{Experimental results}

In this section, we compare the efficiency of the classical deterministic F4 algorithm with the Las-Vegas variant (\Cref{alg:lv-f4}). To this end, we have implemented the Las-Vegas variant in Julia in Groebner.jl~\cite{demin2024groebnerjl}. For measuring running times, we use an Intel Xeon Gold 6538 server, and run all computations on a single thread\footnote{Benchmark script is available at \url{https://github.com/sumiya11/Groebner.jl/blob/v0.10.9/scripts/benchmark_linalg/las_vegas.jl}}.

\begin{table}[H]
\centering
\begin{tabular}{l||r||r|r|r||r}
\multirow{2}{*}{Example}
& \multirow{2}{*}{Deterministic F4}
& \multicolumn{3}{c||}{Las-Vegas F4}
& \multirow{2}{*}{$\dfrac{\text{Las-Vegas F4}}{\text{Deterministic F4}}$}
\\
& & Monte-Carlo F4 & \grobner{} basis test & Total & \\[0.3em]
\hline\hline
{\tt Cyclic-9}               & 250 s & 47 s & 5 s  & 52 s & 0.21\\
{\tt Cholera}~\cite{weights} & 47 s  & 16 s & 0 s  & 16 s & 0.34\\
{\tt Goodwin}~\cite{weights} & 64 s  & 58 s & 0 s  & 58 s & 0.91\\
{\tt Katsura-12}             & 22 s  & 2 s  & 21 s & 23 s & 1.05\\
{\tt Noon-9}                 & 6 s   & 5 s  & 11 s & 16 s & 2.53\\
\end{tabular}
\caption{The running time of the deterministic and the Las-Vegas variants of the F4 algorithm. The Las-Vegas one comprises two sub-algorithms: the Monte-Carlo F4 algorithm and the \grobner{} basis test.}
\label{table:benchmarks1}
\end{table}

We will be computing \grobner{} bases over integers modulo the prime $2^{30}+3$ in degrevlex monomial ordering. We will use the following implementations, which share the main data structures:
\begin{itemize}
    \item Classical deterministic F4 algorithm~\cite{F4}, distributed within Groebner.jl and described in~\cite{demin2024groebnerjl}.
    \item The Las-Vegas variant that follows the spirit of~\Cref{alg:lv-f4}, with practical modifications. Recall that the algorithm consists of the Monte-Carlo variant of the F4 algorithm (\Cref{alg:mc-f4}) and the \grobner{} basis test (\Cref{prop:test-gb}). The Monte-Carlo sub-algorithm implementation differs from deterministic F4 only in the algorithm used for computing row echelon forms. We use the probabilistic row echelon form algorithm (\Cref{alg:randomized-rowspace-first}) with the options $\ell := \lceil\sqrt{N/3}\rceil$ and $s := 1$. 
\end{itemize}

\Cref{table:benchmarks1} includes the running times of these implementations for several problems\footnote{The definitions of the problems are available at \url{https://github.com/sumiya11/Groebner.jl/blob/v0.10.9/src/utils/examples.jl}}. 
From the last column, we observe that the relative performance varies between examples. In {\tt Goodwin}, the deterministic F4 and the Monte-Carlo variant have similar running times, so we do not expect the Las-Vegas algorithm to be beneficial.
Naturally, the cost of the \grobner{} basis test is influenced by the size of the output. For the {\tt Cholera} and {\tt Goodwin} examples, all S-polynomials in the \grobner{} basis test are discarded without performing linear algebra, which makes the cost of this part negligible. At the other extreme, for {\tt Noon-9}, the cost of the \grobner{} basis test exceeds that of the F4 computation: the largest Macaulay matrix in the F4 algorithm involves 6000 S-polynomials, whereas
the \grobner{} basis test considers 29191 S-polynomials.

These experiments suggest that the Las-Vegas variant of the F4 algorithm could be appealing when both of the following hold: probabilistic row echelon form computation provides a speedup over classical F4, and it is inexpensive to test the
resulting \grobner{} basis.

\bibliographystyle{abbrvnat}
\bibliography{bib}

\end{document}